\documentclass[11pt]{article}
\usepackage{amssymb}
\usepackage{amsmath}
\usepackage{titlesec}

\usepackage{etex}
\usepackage{amsmath}
\usepackage{amssymb}
\usepackage{amsthm}
\usepackage{amsfonts}
\usepackage{epsfig}
\usepackage{latexsym}
\usepackage{url}
\usepackage{ifthen}
\usepackage{comment}
\usepackage{enumerate}
\usepackage{multirow}
\usepackage[english]{babel}
\usepackage{texnansi}
\usepackage{textcomp}
\usepackage{color}
\usepackage{calc}
\usepackage[normalem]{ulem}
\usepackage{array}
\usepackage{booktabs}
\usepackage[margin=10pt,font=small,labelfont=bf]{caption}
\usepackage{sectsty}
\usepackage[capitalise]{cleveref}
\usepackage{tikz}
\usepackage{soul}
\usepackage{setspace}
\usepackage[margin=1in]{geometry}
\usepackage{lipsum}

\provideboolean{lucidabr}

\ifthenelse{\boolean{lucidabr}}{%
  \usepackage{lucidabr}
}{%
  \usepackage{lmodern}
}

\allsectionsfont{\large\sffamily}
\makeatletter
\def\@seccntformat#1{\csname the#1\endcsname.\quad}
\makeatother

\usepackage{tikz}
\usepackage{pgfpages}
\usepackage{pgfplots}
\usepackage{pgfplotstable}
\usetikzlibrary{arrows,decorations,patterns,trees,matrix,calc,shapes}

\newcommand{\field}[1]{\mathbb{#1}}
\newcommand{\R}{\field{R}} %
\newcommand{\1}{\mathbf{1}} %
\newcommand{\E}{\mathsf{E}} %

\DeclareMathOperator*{\argmax}{\text{argmax}}

\renewcommand\max{\mathop{\text{max}}\limits}

\renewcommand\lim{\mathop{\text{lim}}\limits}
\renewcommand\log{\mathop{\text{log}}\limits}
\renewcommand\exp{\mathop{\text{exp}}\limits}
\renewcommand\sup{\mathop{\text{sup}}\limits}
\renewcommand\inf{\mathop{\text{inf}}\limits}

\tikzstyle{every picture} += [>=stealth]

\tikzset{axis/.style={semithick, line join=miter}}

\pgfdeclarelayer{background}
\pgfdeclarelayer{foreground}
\pgfsetlayers{background,main,foreground}

\pgfplotstableset{%
  font=\small,
  every head row/.style={before row=\toprule[1pt], after row=\midrule},
  every last row/.style={after row=\bottomrule[1pt]}}

\def\coursename#1{%
  \def\ctemp{#1}%
  \ifx\ctemp\@empty
  \def\insertcoursename{}
  \else
  \def\insertcoursename{\ignorespaces#1}
  \fi
}
\coursename{}



\usepackage{natbib}
 \bibpunct[, ]{(}{)}{,}{a}{}{,}%

\newcommand{\q}[0]{\alpha}

\newcommand{\eps}{\varepsilon } %

\newcommand{\twopartdef}[4]{
 \begin{cases} #1 &\mbox{if } #2\\
#3 &\mbox{if } #4 \end{cases}}
\newcommand{\threepartdef}[6]{
 \begin{cases} #1 &\mbox{if } #2\\
#3 &\mbox{if } #4\\
#5 &\mbox{if } #6 \end{cases}}

\newcommand{\unp}{0}
\newcommand{\ovp}{\overline{\theta}}

\newtheorem{theorem}{\bfseries\sffamily Theorem}

\newtheorem{corollary}{\bfseries\sffamily Corollary}

\newtheorem{proposition}{\bfseries\sffamily Proposition}

\newcommand{\fcadd}[1]{\noindent{\textcolor{black}{#1}}}

\newcommand{\prof}{\pi}
\providecommand{\keywords}[1]{\textit{Keywords:} #1}
\providecommand{\JEL}[1]{\textit{JEL Classification:} #1}
\titleformat*{\section}{\sffamily\Large\bfseries}
\titleformat*{\subsection}{\sffamily\large\bfseries}
\titleformat*{\subsubsection}{\sffamily\large\bfseries}
\input{tcilatex}

\begin{document}

\author{Dirk Bergemann\thanks{
Department of Economics, Yale University, New Haven, U.S.A.,
dirk.bergemann@yale.edu} \and Francisco Castro\thanks{
Anderson School of Management, UCLA, Los Angeles, U.S.A.,
francisco.castro@anderson.ucla.edu} \and Gabriel Weintraub\thanks{
Graduate School of Business, Stanford University, Stanford, U.S.A.,
gweintra@stanford.edu} }
\title{{\LARGE \bfseries\sffamily Third-Degree Price Discrimination Versus
Uniform Pricing}\thanks{%
The first author would like to acknowledge financial support through NSF
grant SES 1459899. We thank John Vickers for helpful
conversations.}}
\date{\ \today }
\maketitle

\abstract{
We compare the profit of the optimal third-degree price
discrimination policy against a uniform pricing policy. A uniform pricing policy
offers the same price to all segments of the market. Our main result
establishes that for a broad class of third-degree price discrimination
problems with concave  profit functions  (in the price space) and common support, a uniform price
is guaranteed to achieve one half of the optimal monopoly profits. This
 profit bound holds for \fcadd{any} number of segments and prices that
the seller \fcadd{might use under} third-degree price
discrimination. We establish that these conditions are tight and
that weakening either common support or concavity can lead to arbitrarily poor
profit comparisons even for regular or monotone hazard rate distributions.
}

\medskip

\medskip \JEL{C72, D82, D83.} \newpage

\section{Introduction}

\subsection{Motivation and Results}

An important use of information about demand is to engage in price
discrimination. A large body of literature, starting with the classic work
of \citet{pigo20}, examines what happens to prices, quantities, and various
measures of welfare as the market is segmented. A seller engages in
third-degree price discrimination if he uses information about consumer
characteristics to offer different prices to different segments of the
market. As every segment is offered a different price, there is scope for
the producer to extract more surplus from the consumer. Additional information about consumer demand increases flexibility
 for market segmentation.\footnote{\citet{pigo20} suggested a
classification of different forms of price discrimination. First-degree (or
perfect) price discrimination is given when the monopolist charges each unit with
a price that is equal to the consumer's maximum willingness to pay for that
unit. Second-degree price discrimination arises when the price depends on
the quantity (or quality) purchased. Third-degree price discrimination
occurs when different market segments are offered different prices, e.g. due
to temporal or geographical differentiation.}

Our main contribution is to compare the profit performance of third-degree
price discrimination against a uniform pricing policy. A uniform pricing
policy offers the same price to all segments of the market. Theorem \ref%
{thm1} establishes that for a broad class of third-degree price
discrimination problems with concave profit functions (in the price space)  and common support, a
uniform price is guaranteed to achieve one half of the optimal monopoly
profits. {The profit bound  in Theorem \ref{thm1} is independent of the number of segments and holds for any arbitrarily large number of segments.
}\label{rev-seg-intro}%
 Interestingly, the performance guarantee of \cref{thm1} can be established
 with different 
choices regarding the uniform price, each of which uses different sources of information  regarding the market demand.

We investigate the limits of this result by weakening the assumptions of
concavity and common support. First, Proposition \ref{prop:no-com-sup} shows
the significance of the common support assumption by studying a setting with
concave profit functions that have finite but different supports. We display
a sequence of segments under which the profit ratio of uniform price to
third-degree price discrimination goes to zero.\footnote{%
In the related literature section we discuss the relation between this
result and \citet {malueg2006bounding}.} Second, Proposition \ref{prop:triangle} and Proposition \ref
{prop:fail-regu} note that the approximation result does not hold for the
commonly studied class of regular distributions. 
More specifically, we can
weaken the concavity of the profit function to merely assume regular
environments while maintaining common support. In other words, we assume that the profit function is only
concave in the space of quantiles, rather than prices. {Proposition \ref{prop:triangle} establishes that for some regular distributions, uniform
pricing can perform arbitrarily poorly compared to optimal third-degree
price discrimination. Proposition \ref%
{prop:fail-regu} establishes that when we consider the even more selective sub-family of monotone hazard rate distributions, the poor performance of uniform pricing still holds. That is, in the third-degree price discrimination setting we need more stringent conditions, such as concavity,
beyond the most commonly used notion of regularity to attain good approximations.
}\label{rev:prop4-prop5}The importance of the aforementioned results is that they establish that
  if any of these assumptions is dropped (not necessarily at the same time), then the profit ratio of uniform price to
third-degree price discrimination can be small. \label{revision-intro-1}

As an application of our main result, we consider the dynamic mechanism design problem of sequential screening with ex-post participation constraints investigated by \citet{krst15} and \citet{bergemann2017scope}. In \cref{revision-seq-screening}, we establish the connection between the aforementioned problem and our setting and show that \cref{thm1} implies that the static mechanism in sequential screening can lead to a half approximation of the optimal dynamic mechanism.

\subsection{Related Literature}

\label{lit}

Our work builds on the classic literature on third-degree price
discrimination, see e.g., \citet{pigo20}, \citet{robi33}, and  \citet{schm81}%
. In more recent work, \citet{agcv10} identify conditions on the shape of the
demand function for price discrimination to increase welfare and output
compared to the non-discriminating price case.

\citet{bebm15} analyze the limits of price discrimination. They show that
the segmentation and pricing induced by the additional information can
achieve every combination of consumer and producer surplus such that:\ $%
\left( i\right) $ consumer surplus is nonnegative, $\left( ii\right) $
producer surplus is at least as high as profits under the uniform monopoly
price, and $(iii)$ total surplus does not exceed the surplus generated by
the efficient trade. Building on this work, \citet{cummings2020algorithmic} 
provide approximate guarantees to segment the market when an intermediary
has only partial information about the buyer's values.

In contrast, in this paper we analyze the \emph{profit} implications of
uniform pricing versus third-degree price discrimination. We are
particularly interested in understanding  the approximation
guarantees that a uniform price can deliver. Closest to our work is a paper
by \citet{malueg2006bounding} which examines the profit effects of
third-price discrimination compared to uniform pricing. They consider a
setting similar to ours in which the monopolist experiences a total cost
function for serving different segments. They show that when the demand is
continuous and the total cost is superadditive, the ratio of third-degree
price discrimination profit to uniform price profit is bounded above by the
number of segments that are served under price
discrimination. They provide an example under which this bound is tight and the bound for the ratio equals the total number of segments.
In contrast, in the present paper we identify a key condition
which leads to a bound that is not contingent on the number of segments in
the market.

{Their Proposition 2, adjusted to our setting, implies our \cref{prop:no-com-sup}. 
They also provide an example that attains the worst-case performance for distributions with different support (and linear demand).
While their proof is inductive, we provide an alternative and constructive argument. }\label{rev:malueg-snyder}

{Since the seminal work of \citet{myer81}, there has been great 
interest in simple and approximate mechanisms
design. In general, characterizing optimal selling mechanisms is a difficult
task, see e.g., \citet{daskalakis2014complexity} and %
\citet{papadimitriou2016complexity}. Hence, deriving simple-practical
mechanisms is of utmost importance.  \citet{chawla2007algorithmic}, \citet{hartline2009simple}, \citet{alaei2019optimal}, and \citet{jin2018tight}, among others, have made remarkable progress toward establishing  performance guarantees of simple mechanisms in a variety of settings. %
One of the key observations is that regular environments---non-decreasing virtual value---consistently lead to good bounds. In particular, triangular instances---instances for which the revenue functions in the \textit{quantile space} are triangle-shaped (see e.g., \citet{alaei2019optimal})---are the worst-case in terms of performance guarantees. In contrast to this stream of literature,  we consider the problem faced by a monopolist
selling to %
many distinct segments of the market. Nevertheless, in line with these earlier papers, we aim to obtain
performance guarantees when comparing the best possible pricing for the
monopolist (third-degree price discrimination) to the simple pricing
scheme (uniform pricing). In terms of techniques, as we discuss in the
next paragraph, we resort to related
triangular instances  as worst-case performance settings, but we also establish arbitrarily poor performance in the regular case.}
\label{revision-lit-1}

Our work also shares some similarities with the approach taken by %
\citet{dhangwatnotai2015revenue} (see also \citet{hartline2013mechanism} chapter 5) to
study the prior-independent single sample mechanism. In this mechanism,
bidders are allocated an object according to the VCG mechanism with reserves randomly
computed from other bidders' bids. 
In a setting with  $n$ bidders, the authors establish that 
 random pricing achieves half of the
optimal profit. %
The authors 
 expand this result to more
complex settings and formalize it by using an intuitive geometric
approach similar to the one we present in Section \ref{sec:concave}. In
particular, under the assumption of regular distributions, the profit
function in the quantile space turns out to be concave. Consequentially, the
profit function is bounded below by a triangle with height equal to the
maximum profit. This implies that the expected profit from uniformly
selecting a quantile is bounded below by the area of the triangle or,
equivalently, by half the maximum profit. One of the proofs we give for our result in Theorem 
\ref{thm1} %
 uses this observation to prove a different
result, namely, that uniform pricing can deliver at least half the value of optimal third-degree price discrimination. However, we must assume that the
profit functions are \emph{concave in the price space}, otherwise our half
approximation result might not hold. Indeed, in Proposition \ref%
{prop:fail-regu} we show that, in contrast to the aforementioned papers in our third-degree price discrimination setting, for regular distributions, simple pricing
leads to arbitrarily poor guarantees.\label{revision-lit-2}

\section{Model}

\label{sec:mod}

We consider a monopolist selling to $K$ different customer segments. Each
segment $k$ is in proportion $\alpha _{k}$ in the market where $\alpha
_{k}\geq 0$ for all $k\in \{1,\dots ,K\}$ and $\sum_{k=1}^{K}\alpha _{k}=1$.
If the monopolist offers price $p_{k}$ to segment $k$ and has constant marginal cost $c\geq 0$, then the monopolist
receives an associated profit of: 
\begin{equation*}
\prof_{k}(p_{k})\triangleq (p_{k}-c)\cdot (1-F_{k}(p_{k})) \ ,
\end{equation*}%
where $F_{k}(\cdot )$ is the cumulative distribution function of a distribution
with support in $\Theta _{k}\subset \R_{+}$.We assume that $c\leq \sup\{\theta:\theta\in \Theta_k\}, \forall k$, i.e.,  the efficient allocation would generate sales with positive probability in every segment.
  The total profit the monopolist
receives from the different segments by pricing according to a vector of prices $\boldsymbol{p}=(p_{1},\dots ,p_{K})$ is 
\begin{equation*}
\Pi (\boldsymbol{p})=\sum_{k=1}^{K}\alpha _{k}\prof_{k}(p_{k}).
\end{equation*}%
The monopolist wishes to choose $\boldsymbol{p}$ to maximize $\Pi (%
\boldsymbol{p})$. 

The monopolist can choose prices in different manners. First, for each segment $%
k$, the monopolist can set the price $p^\star_k$ where 
\begin{equation*}
p^\star_k\in \argmax_{p\in \Theta_k} \prof_k(p).
\end{equation*}
Let $\boldsymbol{p}^\star$ be the vector of prices $\{p^\star_k\}_{k=1}^K$;
we refer to these prices as the \textit{per-segment optimal prices}. Note
that $\boldsymbol{p}^\star$ corresponds to the case of third-degree price
discrimination. We use $\Pi^\star$ to denote $\Pi(\boldsymbol{p}^\star)$. 
Another way of setting prices  is to simply use a
 uniform price for all segments. In this case, the monopolist  solves
the problem 
\begin{equation}\label{rev:uniform-problem}
\Pi ^{U}\triangleq \max_{p\in \cup _{k=1}^{K}\Theta
_{k}}\sum_{k=1}^{K}\alpha _{k}\prof_{k}(p).
\end{equation}
We use $p_u$ to denote the optimal price in the above problem, which we
refer to as the \textit{optimal uniform price}. With some abuse of notation
we sometimes use $\Pi ^{U}(p)$ to denote $\Pi (\boldsymbol{p})$ when all the
components of $\boldsymbol{p}$ are equal to $p$. 
\fcadd{
We call the ratio between  the best third-degree price discrimination scheme
and  the best uniform price scheme the \textit{profit ratio}:
\begin{equation}\label{eq:profit-ratio-rev}
\frac{\Pi ^{U}}{\Pi ^\star}. 
\end{equation}
We use  $\Pi^\star(\alpha,F)$ and $\Pi^U(\alpha,F)$ to make explicit the dependence of the
monopolist profit on the model parameters ($\alpha, F$): the segmentation, $\alpha=\{\alpha_k\}_{k=1}^K$, and the demand, $F=\{F_k\}_{k=1}^K$.
Our main objective in this
paper is to study how this ratio performs across a wide rage of parameter environments:
\begin{equation}
\inf_{\alpha ,F} \frac{\Pi ^{U}(\alpha ,F)}{\Pi ^\star(\alpha ,F)}. 
\tag{$\mathcal{P}$}  \label{eq:min-ratio}
\end{equation} 
}

\section{Concave Profit Functions}
\label{sec:concave}

In this section, we assume that the profit functions, $\prof_{k}(\cdot )$, are
concave. We will further assume that the segments' supports, $\Theta _{k}$,
are identical across segments, that is,  $\Theta _{k}=\Theta $ for all $k$ where $\Theta $
is a closed and bounded interval $[\unp,\ovp]$ of $\R_{+}$. In later
sections, we analyze \eqref{eq:min-ratio} under relaxed assumptions.

We now establish that a particularly simple uniform price, formed as the midpoint between the marginal cost $c$ and 
the largest possible value $\ovp$,
\begin{equation}\label{eq:simple-uniform-price}
p_s  = \frac{c+\ovp}{2},
\end{equation}
 can achieve half of the monopolist's profit. 

\begin{theorem}[\textbf{Uniform price is a half approximation}]\label{thm1}{\ \\}
Suppose that the profit functions $\prof_k(p)$ are concave 
and defined in the same bounded interval $\Theta\subset \R_+$ for all $k\in\{1,\dots, K\}$. 
 Then the uniform
price $p_s$ delivers a 1/2-approximation for the monopolist's profits.
\end{theorem}
\begin{proof}
We show that by setting $p_s = (c+\ovp)/2$ the monopolist can obtain at least half the profit of third-degree price discrimination. Indeed, the concavity of $\prof_k(\cdot)$ ensures that: %
\begin{equation}
\frac{\prof_k(p_s)}{p_s-c}\geq \frac{\prof_k(p^\star_k)}{p^\star_k - c}\:\: \text{ if } p_s\leq p^\star_k,\quad \text{and}
\quad
\frac{\prof_k(p_s)}{\ovp-p_s}\geq \frac{\prof_k(p^\star_k)}{\ovp-p^\star_k}\:\: \text{ if } p_s> p^\star_k.
\end{equation}
By noticing that $(p_s-c)/(p^\star_k - c)\geq 1/2$, $(\ovp-p_s)/(\ovp-p^\star_k)\geq 1/2$, and $p_k^*\in[c,\ovp]$,
we deduce that in either case $\prof_k(p_s)\geq \prof_k(p^\star_k)/2$ for all $k\in\{1,\dots,K\}$. Multiplying this inequality by $\alpha_k$, adding it up over $k$ , and observing that $\Pi^U\geq\Pi^U(p_S)$, we obtain the desired result: 
$$\inf_{\alpha ,F}\frac{\Pi ^{U}(\alpha ,F)}{\Pi ^\star(\alpha ,F)} \geq 1/2.$$
\end{proof}
Theorem \ref{thm1} provides a fundamental guarantee of uniform
pricing compared to the optimal third-degree price discrimination. In
particular, the monopolist can  simply use a judiciously chosen price across
all customer segments to ensure half of the best possible profit from
perfectly discriminating across the different segments in the market. 
Interestingly, it is possible to achieve this profit guarantee by 
setting a price  that only uses the upper bound of the support and the marginal cost.\footnote{We thank John Vickers for this suggestion.} An implication of this is that half of the optimal profit 
can be secured by using only information about the upper end of segments' support, $\ovp$. While the simplicity of this pricing policy is appealing, the informational requirements might be too stringent. In practice, finding the upper bound of the support can be difficult because it can entail experimenting with high prices. 
We next investigate what other simple pricing policies can achieve the approximation guarantee in \cref{thm1} and discuss their informational requirements.

 In what follows, we provide a simple geometric argument as an alternative to the proof of \cref{thm1}. 
This alternative proof not only  sheds light on different informational requirements of simple pricing policies that can achieve at least 1/2 performance, but also shows how ideas used in approximate mechanism design translate to our setting. Indeed, the next argument is similar to the one presented in %
\citet{dhangwatnotai2015revenue} and \citet{hartline2013mechanism} for concave profit functions in quantile space, i.e., for regular distributions. %

Let 
\begin{equation*}
r_k\triangleq \q_k\prof_k(p^\star_k),
\end{equation*}
that is, $r_k$ corresponds to the maximum profit the seller can obtain from
the fraction $\alpha_k$ of  customers in segment $k$. Note that $\Pi^\star$ equals $%
\sum_{k=1}^Kr_k$. Since for each segment $k$ the profit function is concave
in $\Theta$, we can lower bound it by a triangular-shaped function that we
denote by $L_k(p)$ as depicted in Figure \ref{fig:density1} (a). 
\begin{figure}[]
\centering
\scalebox{0.675}{\begin{tikzpicture}[baseline=0pt,scale=0.6]
\def\n{9}
\def\dl{-16.5}
\draw [->,black,line width=0.8mm] (0,0) -- (1.7*\n,0);
\draw [->,black,line width=0.8mm] (0,0) -- (0,0.8*\n);
\draw node at (1.7*\n+0.1,-0.8) {\large $p$};
\draw [-,line width=0.8mm] (0,0.3) -- (0,-0.3) node at (0,-0.8) { $c$};

\draw [-,line width=0.8mm]  node at (-1.0,7.75) {\large $\sum_{k=1}^K L_k(p)$};

\draw [-,line width=0.8mm] (1.5*\n,0.3) -- (1.5*\n,-0.3) node at (1.5*\n,-0.8) {\large $\ovp$};
\draw [-,line width=0.8mm] (4.4010,0.3) -- (4.4010,-0.3) node at (4.4010,-0.8) {\large $p^\star_1$};
\draw [dashed,-,line width=0.4mm] (4.4010,0) -- (4.4010,3.95) node at (4.8,3.6) {};
\draw [-,line width=0.8mm] (6.75,0.3) -- (6.75,-0.3) node at (6.75,-0.8) {\large $p^\star_2$};
\draw [dashed,-,line width=0.4mm] (6.75,0) -- (6.75,5.0) node at (7.2,4.6) {};
\draw [-,line width=0.8mm] (10.0710,0.3) -- (10.0710,-0.3) node at (10.0710,-0.8) {\large $p^\star_3$};
\draw [dashed,-,line width=0.4mm] (10.0710,0) -- (10.0710,4.3) node at (9.7,3.95) {};

plot curves 
\draw[dashed,black,line width=0.6mm,shift={(0,0.1)}] plot file {TriangleRight1.data};
\draw[black,line width=0.6mm,shift={(0,0.1)}] plot file {TriangleRight2.data};

\draw[dashed,black,line width=0.4mm] (0.1,5.55)--(8.0,5.55) node at (-0.65,5.55) {\large $\Pi^U$};
\draw[dashed,black,line width=0.4mm] (8.0,0)--(8.0,5.55) node at (8.0,-0.7){$p^U$};
\draw [-,line width=0.8mm] (8.0,0.3) -- (8.0,-0.3);
\draw[dashed,black,line width=0.4mm] (0.1,5.0)--(7.0,5.0) node at (-0.65,5.0) {\large $\Pi^L$};

\draw [dashed,-,line width=0.4mm] (4.4010+\dl,0) -- (4.4010+\dl,2.0);
\draw [-,line width=0.8mm] (4.4010+\dl,0.3) -- (4.4010+\dl,-0.3) node at (4.4010+\dl,-0.8) {\large $p^\star_1$};
\draw [dashed,-,line width=0.4mm] (6.75+\dl,0) -- (6.75+\dl,3.5);
\draw [-,line width=0.8mm] (6.75+\dl,0.3) -- (6.75+\dl,-0.3) node at (6.75+\dl,-0.8) {\large $p^\star_2$};
\draw [-,line width=0.8mm] (10.0710+\dl,0.3) -- (10.0710+\dl,-0.3) node at (10.0710+\dl,-0.8) {\large $p^\star_3$};
\draw [dashed,-,line width=0.4mm] (10.0710+\dl,0) -- (10.0710+\dl,4.1);

\draw [->,black,line width=0.8mm] (0+\dl,0) -- (1.7*\n+\dl,0);
\draw [->,black,line width=0.8mm] (0+\dl,0) -- (0+\dl,0.8*\n);
\draw node at (1.7*\n+0.1+\dl,-0.8) {\large $p$};
\draw [-,line width=0.8mm] (0+\dl,0.3) -- (0+\dl,-0.3) node at (0+\dl,-0.8) { $c$};
\draw [-,line width=0.8mm]  node at (-1.0+\dl,7.5) {\large $\prof_{k}(p)$};
\draw [-,line width=0.8mm] (1.5*\n+\dl,0.3) -- (1.5*\n+\dl,-0.3) node at (1.5*\n+\dl,-0.8) {\large $\ovp$};

\draw[black,line width=0.6mm,shift={(0+\dl,0.1)}] plot file {TriangleLeft1-C.data};
\draw[black,line width=0.6mm,shift={(0+\dl,0.1)}] plot file {TriangleLeft2-C.data};
\draw[black,line width=0.6mm,shift={(0+\dl,0.1)}] plot file {TriangleLeft3-C.data};
\draw[black,dashed,line width=0.6mm,shift={(0+\dl,0.1)}] plot file {TriangleLeft1-T.data};
\draw[black,dashed,line width=0.6mm,shift={(0+\dl,0.1)}] plot file {TriangleLeft2-T.data};
\draw[black,dashed,line width=0.6mm,shift={(0+\dl,0.1)}] plot file {TriangleLeft3-T.data};

\node at (\n*0.8,-2.2){\large \textbf{(b)}};
\node at (\n*0.8+\dl,-2.2){\large \textbf{(a)}};
\end{tikzpicture}}
\caption{\textbf{(a)} The solid curves depict the concave profit function of
each segment, $\q_k\prof_k(p)$. The dashed lines depict the lower bounds $%
L_k(p) $ for each  segment. \textbf{(b)} The solid curve shows the sum of the
profit functions over segments, $\sum_{k=1}^K\protect\alpha_k \prof_{k}(p)$. The
dashed  curve shows the sum of the lower bound over  segments, $%
\sum_{k=1}^K L_{k}(p)$.}
\label{fig:density1}
\end{figure}
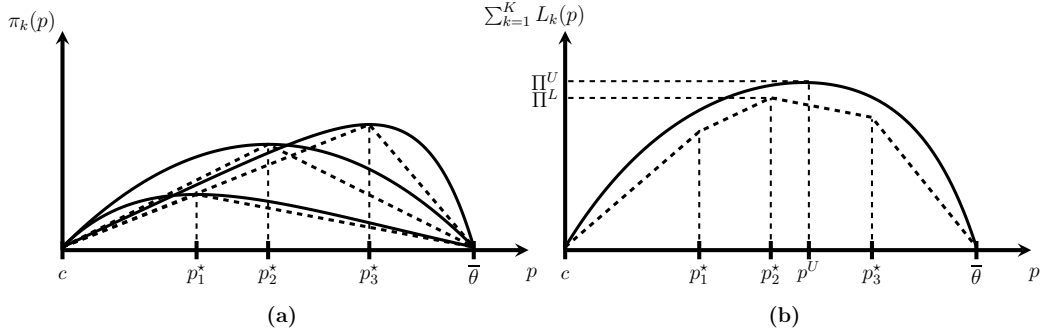

More precisely, we define the lower bound functions 
\begin{equation*}
L_k(p)\triangleq \twopartdef{\frac{r_k}{p^\star_k-c}\cdot
(p-c),}{p\in[c,p^\star_k];}{\frac{r_k}{\ovp-p^\star_k}\cdot(\ovp-p),}{p\in[p^%
\star_k,\ovp].}
\end{equation*}%
Observe that $\sum_{k=1}^{K}L_k(p)$ is a concave piecewise linear
function that achieves its maximum at some $p\in \{p_{k}^{\star }\}_{k=1}^{K}$.
We use $\Pi ^{L}$ to denote its maximum value. Then, it is easy to see that $%
\Pi ^{U}\geq \Pi ^{L}$ because $\sum_{k=1}^{K}L_k(p)$ lower bounds $%
\sum_{k=1}^{K}\alpha _{k}\prof_{k}(p)$, see Figure \ref{fig:density1} (b).  Next,
we argue that 
\begin{equation}
\Pi ^{L}=\max_{p\in \{p_{1}^{\star },\dots ,p_{K}^{\star }\}}\left\{
\sum_{k=1}^{K}L_k(p)\right\} \geq \frac{1}{2}\sum_{k=1}^{K}r_{k}=\frac{%
1}{2}\Pi^\star.  \label{eq:geom-1}
\end{equation}%
 Consider Figure \ref%
{fig:gem-1} and note that $\Pi ^{L}\cdot (\ovp-c)$ is equal to the area of the
smallest rectangle that contains the graph of $\sum_{k=1}^{K}L_k(p)$.
As a consequence, $\Pi ^{L}\cdot (\ovp-c)$ is an upper bound for the area below
the curve $\sum_{k=1}^{K}L_k(p)$. That is, 
\begin{equation*}
\Pi ^{L}\cdot (\ovp-c)\geq \int_{c}^{\ovp}\sum_{k=1}^{K}L_k(p)dp=%
\sum_{k=1}^{K}\int_{c}^{\ovp}L_k(p)dp=\sum_{k=1}^{K}\frac{r_{k}\cdot %
(\ovp-c)}{2},
\end{equation*}%
where in the last equality we have used the fact that $L_k(p)$ is
triangle-shaped and, therefore, the area below its curve equals $r_{k}\cdot %
(\ovp-c)/2$. Dividing both sides in the expression above by $(\ovp-c)$ yields
\eqref{eq:geom-1}, completing the proof.   In particular, 
\begin{equation}\label{eq:rev-two}
\Pi^{U} = \max_{p\in\Theta}\left\{\sum_{k=1}^{K}\q_k \prof_k(p)\right\}
\geq \max_{p\in\{p^\star_1,\dots, p^\star_K\}}\left\{\sum_{k=1}^{K}\q_k \prof_k(p)\right\}\geq \frac{1}{2}\Pi^\star.
\end{equation}
This argument  suggests two distinct yet simple  ways of selecting the uniform price.
First, the monopolist can optimize against the mixture of customer segments
to derive the optimal uniform price. This is advantageous for situations in
which the monopolist possesses aggregate market information \fcadd{but} 
discriminating across segments is not an available option. When the monopolist has
more granular market information, for example, the monopolist knows the
prices $\{p_{k}^{\star }\}_{k=1}^{K}$, then it is not necessary for the
monopolist to optimize over the full range of prices; he can simply
choose one of the $K$ prices at hand. 

\label{rev:sum-inf-2}
In sum, the two distinct arguments presented for \cref{thm1} complement each other and point to different ways  the monopolist has of achieving the performance guarantee of 1/2, depending on the information available. If the monopolist knows  the upper bound of the support, setting $p_s=(c+\ovp)/2$ is a simple choice. However, if such information is not available then there are two more  options \fcadd{that can be directly inferred from the inequalities \eqref{eq:rev-two}: {\it{(i)}} the monopolist can find the optimal uniform price for the aggregate demand or {\it{(ii)}} identify among the best per segment prices the uniform price that maximizes the revenue from the aggregate demand.}

\begin{figure}[]
\centering
\scalebox{0.8}{\begin{tikzpicture}[baseline=0pt,scale=0.6]
\def\n{9}
\def\dl{0}
\draw [->,black,line width=0.8mm] (0,0) -- (1.7*\n,0);
\draw [->,black,line width=0.8mm] (0,0) -- (0,1*\n);
\draw node at (1.7*\n+0.1,-0.8) {\large $p$};
\draw [-,line width=0.8mm] (0,0.3) -- (0,-0.3) node at (0,-0.8) { $c$};

\draw [-,line width=0.8mm] (1.5*\n,0.3) -- (1.5*\n,-0.3) node at (1.5*\n,-0.8) {\large $\ovp$};
\draw [-,line width=0.8mm] (4.4010,0.3) -- (4.4010,-0.3) node at (4.4010,-0.8) {\large $p^\star_1$};
\draw [dashed,-,line width=0.4mm] (4.4010,0) -- (4.4010,5.15) node at (4.8,3.6) {};
\draw [-,line width=0.8mm] (6.75,0.3) -- (6.75,-0.3) node at (6.75,-0.8) {\large $p^\star_2$};
\draw [dashed,-,line width=0.4mm] (6.75,0) -- (6.75,6.75) node at (7.2,4.6) {};
\draw [-,line width=0.8mm] (10.0710,0.3) -- (10.0710,-0.3) node at (10.0710,-0.8) {\large $p^\star_3$};
\draw [dashed,-,line width=0.4mm] (10.0710,0) -- (10.0710,6.1) node at (9.7,3.95) {};

plot curves 
\draw[black,line width=0.6mm,shift={(0,0.1)}] plot file {TriangleRight1-geom.data};

\draw [dashed,black,line width=0.4mm] (0,7)--(1.5*\n,7) node at (-0.65,7) {\Large $\Pi^L$};
\draw [dashed,black,line width=0.4mm] (1.5*\n,7)--(1.5*\n,0);

\draw [->,line width=0.4mm] (8,6.65) to [out=45,in=180] (10,8) node at (12.0,8) {$\sum_{k=1}^K L_{k}(p)$};

\draw [dashed,->,black,line width=0.4mm] (4.4010,1.9)--(1.5*\n+1.5,1.9) node at (1.5*\n+1.5+0.5,1.9) {$r_1$};
\draw [dashed,->,black,line width=0.4mm] (6.75,3.5)--(1.5*\n+1.5,3.5) node at (1.5*\n+1.5+0.5,3.5) {$r_2$};
\draw [dashed,->,black,line width=0.4mm] (10.0710,4.15)--(1.5*\n+1.5,4.15) node at (1.5*\n+1.5+0.5,4.15) {$r_3$};

\draw[black,line width=0.6mm,shift={(0+\dl,0.1)}] plot file {TriangleLeft1-T.data};
\draw[black,line width=0.6mm,shift={(0+\dl,0.1)}] plot file {TriangleLeft2-T.data};
\draw[black,line width=0.6mm,shift={(0+\dl,0.1)}] plot file {TriangleLeft3-T.data};

\end{tikzpicture}}
\caption{ Alternative geometric proof of Theorem \protect\ref{thm1}: the area below $%
\sum_{k=1}^{K}L_k(p)$ equals the sum of the areas below $L_k(p)$ for
all $k$.}
\label{fig:gem-1}
\end{figure}
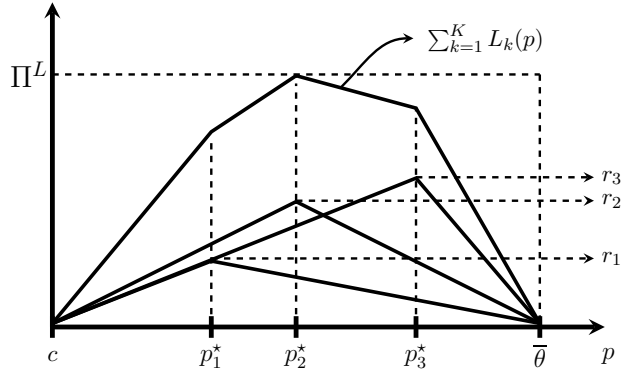

As mentioned in the introduction, our result and our second approach share some
similarities with \citet{dhangwatnotai2015revenue}. For the case of
one buyer---and regular distributions---they show that the expected profit
of randomly selecting a price achieves half of the optimal profit. Their
approach uses the fact that the profit function in the quantile space for
regular distributions is concave, and then proposes a uniform randomization over
quantities. We can use a similar argument to show that the expected
profit of uniformly choosing prices achieves half the profit of third-degree
price discrimination. Indeed, suppose we set a price $p$ at random such that 
$p\sim U[c,\ovp]$. Then, the expected profit is
\begin{equation}  \label{eq:uniform-price}
\E[\Pi^{U}(p)]=\int_{c}^{\ovp} \sum_{k=1}^{K}\q_k \prof_k(p)\cdot \frac{1}{%
\ovp-c}\:dp \geq\frac{1}{\ovp-c}\cdot \int_{c}^{\ovp} \sum_{k=1}^{K}
L_k(p) \:dp = \frac{1}{\ovp-c}\cdot \sum_{k=1}^K \frac{r_k\cdot (\ovp-c)}{2}= 
\frac{1}{2}\Pi^\star.
\end{equation}
We summarize this discussion in the following proposition. 
\label{revision-thm1-2}
\begin{proposition}[\textbf{Uniformly at random pricing}]\label{prop1}{\ \\}
Suppose that the profit functions $\prof_k(p)$ are concave 
and  have common and compact support.
Then for $p\sim U[c,\ovp]$ we have that $\E_p[\Pi^{U}(p)]$ is at least half as large as $\Pi^\star$.
\end{proposition}

To conclude this section, we note that the above pricing policies all correspond to specific instances of simple pricing. Hence, it is still possible that the optimal uniform pricing achieves a better performance than the one established in \cref{thm1}.  In the next proposition, we show that the latter is not possible by establishing that the profit guarantee in \cref{thm1} is tight. To see why this is true, consider Figure \ref{fig:tight-1}. There are two segments in the same proportion with maximum profit equal to 1. Assume that the profit of the first segment is very low at the price of the second segment and vice-versa. Then uniform pricing will achieve only the maximum profit of one of the segments, but very little of the profit from the other segment. In the figure, the best uniform pricing is $p_1^\star$ and it achieves $\frac{1}{2}\cdot 1+\frac{1}{2}\cdot \eps$; while perfect price discrimination achieves $\frac{1}{2}\cdot 1+\frac{1}{2}\cdot1=1$. As $\eps$ becomes small, the ratio of the profits of  optimal uniform pricing to third-degree price discrimination approaches $1/2$. In the proof of \cref{prop:tight-half}, we carefully construct the cumulative distribution functions to mimic the behavior in the above illustration.

\begin{figure}[]
\centering
\scalebox{0.8}{\begin{tikzpicture}[baseline=0pt,scale=0.6]
\def\n{9}
\def\dl{0}
\draw [->,black,line width=0.8mm] (0,0) -- (1.7*\n,0);
\draw [->,black,line width=0.8mm] (0,0) -- (0,1*\n);
\draw node at (1.7*\n+0.1,-0.8) {\large $p$};
\draw [-,line width=0.8mm] (0,0.3) -- (0,-0.3) node at (-0.25,-0.8) { $0$};

\draw [-,line width=0.8mm] (1.5*\n,0.3) -- (1.5*\n,-0.3) node at (1.52*\n,-0.8) {\large $\ovp$};

\draw [-,line width=0.8mm] (0.38,0.3) -- (0.38,-0.3) node at (0.45,-0.8) {\large $p^\star_1$};
\draw [dashed,-,line width=0.4mm] (0.38,0) -- (0.38,7.1) ;

\draw [-,line width=0.8mm] (13.1,0.3) -- (13.1,-0.3) node at (13.0,-0.8) {\large $p^\star_2$};
\draw [dashed,-,line width=0.4mm] (13.1,0) -- (13.1,7.1) ;

\draw [dashed,black,line width=0.4mm] (0,7.15)--(1.47*\n,7.15) node at (-0.65,7.15) {\Large $1$};

\draw [line width=0.4mm,decorate,decoration={brace,amplitude=8pt,mirror,raise=2pt},yshift=0pt]
(0.35,0) -- (0.35,1.3) node [black,midway,xshift=0.5cm] {\footnotesize
$\eps$};

\draw[black,line width=0.6mm,shift={(0+\dl,0.1)}] plot file {tight-R1.data};
\draw[black,line width=0.6mm,shift={(0+\dl,0.1)}] plot file {tight-R2.data};

\end{tikzpicture}}
\caption{Illustration of the construction in  \cref{prop:tight-half}. Profit functions for both segments.}
\label{fig:tight-1}
\end{figure}
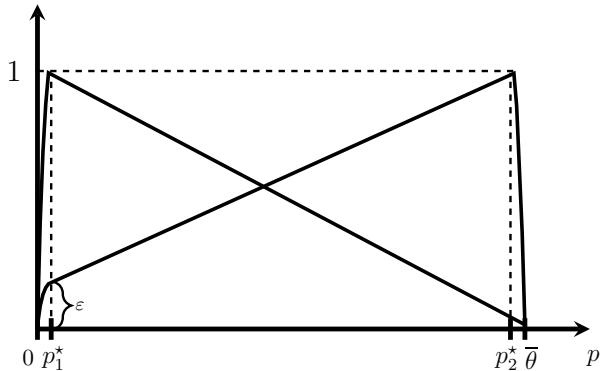

\newpage
\begin{proposition}[\textbf{Tightness}]\label{prop:tight-half}{\ \\}
Under the assumptions of \cref{thm1}
\begin{equation*}
\inf_{\alpha ,F}\frac{\Pi ^{U}(\alpha ,F)}{\Pi^\star(\alpha ,F)}=\frac{1}{2}.
\end{equation*}
\end{proposition}

\begin{proof}Without loss of generality we take $c=0$.
Consider a case with two segments in the same proportion and the following distributions: 
\begin{equation*}
1-F_1(p) \triangleq \twopartdef{\frac{1}{\lambda_1 p}\left(1-e^{-\lambda_1 p}\right),}{p\leq a;}{\frac{1}{p}\frac{(M-p)}{(M-a)},}{p>a,}
\end{equation*}
and 
\begin{equation*}
1-F_2(p) \triangleq \threepartdef{\frac{1}{\lambda_2 p}\left(1-e^{-\lambda_2 p}\right),}{p\leq a;}{\frac{1}{p}\left\{\frac{1-\left(\frac{a}{M-a}+\eps\right)}{(M-2a)}(p-a) + \left(\frac{a}{M-a}+\eps\right)\right\},}{p\in(a,M-a);}{
\frac{1}{p}\left\{
-\frac{1}{a}(p-M+a)+1
\right\},
}{p\in [M-a,M],}
\end{equation*}
with support $\Theta = [0,M]$ for some $M>2a$ to be determined, and we choose $\lambda_k>0$ such that $F_k(\cdot)$ is continuous at $a>1$, and  $\eps>0$ is a small parameter. 
Below, we will show how to choose the parameters to ensure that 
$F_k(\cdot)$ is a well-defined continuous distribution for $k\in\{1,2\}$. First, we  verify that $\prof_k$ is concave for $k\in\{1,2\}$. Indeed, for $k=1$ we have 
\begin{equation*}
\prof_1(p) \triangleq \twopartdef{\frac{1}{\lambda_1}\left(1-e^{-\lambda_1 p}\right),}{p\leq a;}{\frac{(M-p)}{(M-a)},}{p>a.}
\end{equation*}
The first piece of $\prof_1(\cdot)$ is increasing and concave while the second piece is decreasing and linear. The continuity of $F_1$ requires that $\prof_1(a)=1$, equivalently, $(1-e^{-a\lambda_1})/(a\lambda_1) = 1/a$. The latter always has a solution because the function $(1-e^{-x})/x$ maps to the entire interval $(0,1]$ and we are assuming $a>1$. We conclude that $\prof_1(\cdot)$ is concave in $\Theta$. Additionally, by taking the derivative of $1-F_1(\cdot)$, it is possible to verify that $1-F_1(\cdot)$ is decreasing as long as $M>a$.

We now verify the concavity of $\prof_2(\cdot)$. We have
\begin{equation*}
\prof_2(p) \triangleq \threepartdef{\frac{1}{\lambda_2}\left(1-e^{-\lambda_2 p}\right),}{p\leq a;}{\frac{1-\left(\frac{a}{M-a}+\eps\right)}{(M-2a)}(p-a) + \left(\frac{a}{M-a}+\eps\right),}{p\in(a,M-a);}{
-\frac{1}{a}(p-M+a)+1,
}{p\in [M-a,M].}
\end{equation*}
Note that as in the previous case, each piece of $\prof_2(\cdot)$ is concave. We can also verify that it is continuous. This is clear at $p=M-a$, and for $p=a$ we must choose $\lambda_2$ to ensure continuity. In this case, we must have $\prof_2(a)=\eps+a/(M-a)$; equivalently, $(1-e^{-a\lambda_2})/(a\lambda_2) = \eps/a +1/(M-a)$. Assuming that 
\begin{equation}\label{ex1:cond-1}
 \eps/a +1/(M-a)<1,
 \end{equation}
we can always find $\lambda_2$ that makes $\prof_2(\cdot)$ continuous at $p=a$. In turn, $\prof_2(\cdot )$ linearly decreases in $[M-a,M]$, linearly increases in $(a,M-a]$ (because of \eqref{ex1:cond-1}), and increases as a strictly concave function in $[0,a]$. Thus to ensure concavity we need to verify that the slope of $\prof_2(p)$ from the left of $p=a$ is larger than the slope from the right at $p=a$, that is,  
  \begin{equation}\label{ex1:cond-2}
  e^{-\lambda_2 a} \geq \frac{1-\left(\frac{a}{M-a}+\eps\right)}{(M-2a)}. 
  \end{equation}
To see why \eqref{ex1:cond-2} holds, note first that $\lambda_2$, which 
solves $(1-e^{-a\lambda_2})/(a\lambda_2) = \eps/a +1/(M-a)$,  is bounded above by the solution to 
$1/(a\tilde{\lambda}_2) = \eps/a +1/(M-a)$, that is, $\tilde{\lambda}_2\geq \lambda_2$. Hence to verify 
\eqref{ex1:cond-2} it suffices to show that 
  \begin{equation}\label{ex1:cond-3}
  \exp\left(-\frac{1}{ \eps/a +1/(M-a)} \right) \geq \frac{1-\left(\frac{a}{M-a}+\eps\right)}{(M-2a)}. 
  \end{equation}
Let us choose $M$ such that $\eps = 1/(M-a)^{\kappa}$ with $\kappa>0$ to be determined, that is,  $M-a = \eps^{-1/\kappa}$. Then \eqref{ex1:cond-3} becomes
  \begin{equation*}
  \exp\left(-\frac{1}{ \eps/a +\eps^{1/\kappa}} \right) \geq \frac{1-\left(a \eps^{1/\kappa} +\eps\right)}{(\eps^{-1/\kappa}-a)}. 
  \end{equation*}
  Note that for $\eps<1$, as $\kappa\downarrow 0$ the left-hand side converges to $e^{-a/\eps}$ while the right-hand side converges to 0. Hence we can always choose $\kappa\in (0,1)$ such that  \eqref{ex1:cond-2} holds. Moreover, since $\eps^{1/\kappa}<\eps$ (because $\kappa<1$), we have that \eqref{ex1:cond-1} is always satisfied for $\epsilon$ small enough. In conclusion, we can always choose 
  the parameter of our instance such that $\prof_2(\cdot)$ is concave. Additionally, by taking the derivative of $F_2(\cdot)$, it is possible to verify that it is decreasing as long as $M>2a$ (which can also be achieved for $\kappa\in(0,1)$ small).

Under perfect price discrimination we have that $\prof_k(p_k^\star)=1$ for $k\in\{1,2\}$, so that $\Pi^\star=1$. Now, because both  $\prof_1(\cdot)$ and $\prof_2(\cdot)$ increase in $[0,a]$ and decrease in $[M-a,a]$, the optimal uniform price must lie in $[a,M-a]$. In this interval both functions are linear and the derivative of their sum is $-\eps/(M-2a)<0$. Hence the optimal uniform price is $p_u=a$, which yields 
\begin{equation*}
\Pi^U =\frac{1}{2}\cdot 1+\frac{1}{2}\cdot \left( a\eps^{1/\kappa}+ \eps\right)\leq \frac{1}{2} +\frac{1}{2}\eps(a+1).
\end{equation*}
This implies that 
\begin{equation*}
\inf_{\alpha ,F}\frac{\Pi ^{U}(\alpha ,F)}{\Pi^\star(\alpha ,F)}\leq \frac{1}{2} +\frac{1}{2}\eps(a+1).
\end{equation*}
Because $\eps>0$ is arbitrary, we conclude that the lower bound performance of $1/2$ in \cref{thm1} is tight. 
\end{proof}

We note that \cref{prop:tight-half} and \cref{thm1} imply that the ratio in \eqref{eq:min-ratio} does not degrade too fast in the number of segments.
Indeed, the worst performance can be achieved for the case $K=2$, but as $K$ increases the performance does not continue to degrade. This is a consequence of the concavity assumption. If we did not assume concavity we could have profit functions that take values close to zero around the optimal prices for other segments and a value of, for example, 1 at their own optimal per-segment prices. In this case, the ratio in \eqref{eq:min-ratio} would degrade at rate $1/K$ (see also \cref{sec:wcp}).

\section{Profit Performance in General Environments}

\label{sec:other-env} In this section we examine 
weaker conditions relative to the
environment studied in Section \ref{sec:concave}. In particular, we aim to
understand how the profit ratio behaves when we relax the assumptions in
Theorem \ref{thm1}. We first consider the assumption of compact support and then
consider different supports across customer segments. Then, we study
non-concave environments. In the latter, we are especially interested in
common well-behaved environments such as regular and  monotone hazard
rate (MHR) value distributions.  For ease of exposition and without loss of generality in this section and the following we will assume $c=0$, unless otherwise stated.

\subsection{Concave with Unbounded Support}

In Theorem \ref{thm1}, we considered concave profit functions supported on
some common finite interval $\Theta $. In the next proposition, we relax the
finite support assumption while keeping a common support and concave profit
functions across customer segments. 

\begin{proposition}[\textbf{Zero profit gap with unbounded support}]\label{prop:no-gap-com-sup}{\ \\}
Suppose that the profit functions for all segments are concave with common and unbounded support $\Theta=\R_+$. Then $\Pi^{U}=\Pi^\star$. 
\end{proposition}%
\begin{proof}
Without loss of generality, assume that $p_1^\star\le \cdots\le p_K^\star$.
Note that the concavity of the profit functions together with the unbounded support assumption causes each $\prof_k(p)$
to be increasing up to $p_k^\star$ and  then constant and equal to $\prof_k(p_k^\star)$ for any price $p$ larger than
$p_k^\star$ for all $k\ge 1$. This leads to the following distributions:
\label{pg:rev-no-bound-sup}
\begin{equation*}
1- F_k(p)=\twopartdef{ h_k(p)/p,}{p\leq p_k^\star;}{\prof_k(p_k^\star)/p,}{p\geq p_k^\star,}
\end{equation*}
for some increasing and concave function $h_k(\cdot)$ such that $\lim_{p\rightarrow 0}h_k(p)/p =1$, $h_k(p)/p$ is decreasing, and $h_k(0)=0$.
In turn, by setting $p_u$ equal to $p_K^\star$, the profit $\Pi^{U}$ becomes $\sum_{k=1}^K \alpha_k\prof_k(p_k^\star)=\Pi^\star$.
\end{proof}The proposition establishes that in the concave case with
unbounded support there is no gap in the profit between no price discrimination
and full price discrimination. The intuition behind Proposition \ref%
{prop:no-gap-com-sup} is simple. Concavity, together with the unbounded
support assumption, implies that the marginal profit for each segment must
equal zero for sufficiently large prices. As a consequence, setting an equal
and sufficiently large price for every segment achieves the optimal
third-degree price discrimination outcome.

\subsection{Significance of Common Support}

Here we consider concave profit functions supported on some finite interval $%
\Theta_k$ for each segment $k\ge 1$. In contrast to the previous section,
we will not assume that $\Theta_k=\Theta$ for all segments.
{In order to gain intuition, note that, for an arbitrary distribution, the optimal revenue can be arbitrarily small compared to the expected surplus. Indeed, consider for example the distribution $F(v)=1-1/(1+v)$ for $v\ge0$.\footnote{\cite{hartline2009simple} use this distribution to provide a lower bound
on the worst revenue performace ratio of Vickrey with duplicated bidders and the optimal auction.} For this distribution, the optimal revenue  $\sup_{v}\{ v\cdot (1-F(v))\}$ is 1 while the expected surplus is $\infty$.  Hence, if we consider segments with point-mass distributions for every value $v$, the monopolist profit under uniform pricing (which would correspond to the optimal revenue for $F$) can be arbitrarily bad compared to optimal third-degree price discrimination---which would correspond to the expected surplus for $F$.
In the next result, we make this intuition precise  and show that there is a discrete collection of non-point-mass distributions with non-common support for which uniform pricing delivers arbitrarily small profit as the number of  segments increases. 
\label{revision-no-common-sup-1}

\begin{proposition}[\textbf{Significance of common support}]\label{prop:no-com-sup}{\ \\}
\fcadd{If the segments can have distinct supports then the optimal uniform price may yield an arbitrarily small profit ratio as the number of segments increases.}
\end{proposition}%
\begin{proof} We construct concave profit functions with finite support such that $\Theta_k\neq \Theta_j$ for all $k\neq j$. Let the 
distributions $\{F_k(\cdot)\}_k$ be defined by 
\begin{equation*}
F_k(p)  = \twopartdef{0,}{p\in [0,v_k];}{\frac{(p-v_k)(v_k+\eps_k)}{\eps_k p},}{p\in (v_k,v_k+\eps_k],}
\end{equation*}
with 
\begin{equation*}
v_k= \frac{1}{(K-k+1)},\quad  \q_k=\frac{1}{K},\quad \eps_k\in(0,v_{k+1}-v_{k}), \quad \forall k\in\{1,\dots,K\}.
\end{equation*}
This leads to the following profit functions:
\begin{equation*}
\prof_k(p)  = \twopartdef{p,}{p\in [0,v_k];}{\frac{v_k}{\eps_k}(v_k+\eps_k-p),}{p\in (v_k,v_k+\eps_k].}
\end{equation*}
The perfect price discrimination profit is 
\begin{equation*}
\Pi^\star=\sum_{k=1}^K\q_k v_k=\frac{1}{K}\sum_{k=1}^K\frac{1}{K-k+1}=\frac{1}{K}\sum_{k=1}^K\frac{1}{k}.
\end{equation*}
The optimal uniform price must be achieved at one of the $v_k$, hence
\begin{equation*}
\Pi^{U}=\frac{1}{K}\max_{k=1,\dots,K}\Big\{ \sum_{j=k}^Kv_k\Big\}
=\frac{1}{K}\max_{k=1,\dots,K}\Big\{\frac{1}{K-k+1}\cdot (K-k+1)\Big\}=\frac{1}{K}.
\end{equation*}
Hence, 
\begin{equation*}
\frac{\Pi^{U}}{\Pi^\star} = \frac{\frac{1}{K}}{\frac{1}{K}\sum_{k=1}^K\frac{1}{k}}=
 \frac{1}{\sum_{k=1}^K\frac{1}{k}}\approx \frac{1}{\log(K)}\rightarrow 0 \quad \text{as } K\uparrow \infty.
\end{equation*}
\end{proof}

\begin{figure}[h]
\centering
\scalebox{0.8}{\begin{tikzpicture}[baseline=0pt,scale=0.6]
\def\n{9}
\def\dl{16.5}

\draw [->,black,line width=0.8mm] (0,0) -- (10,0);
\draw [->,black,line width=0.8mm] (0,0) -- (0,10);
\draw node at (10,-0.8) {\large $p$};
\draw [-,line width=0.8mm] (0,0.3) -- (0,-0.3) node at (0,-0.8) { $\unp$};
\draw [-,line width=0.8mm]  node at (-1.0,10.5) {\large $\prof_{k}(p)$};
\draw [-,line width=0.8mm] (9,0.3) -- (9,-0.3) node at (9,-0.8) {\large $1$};

\draw[-,line width=0.4mm](0,0)--(9,9);

\draw[-,line width=0.8mm](1.8,0.2)--(1.8,-0.2)node at (1.8,-0.75){\large $\frac{1}{5}$};
\draw[line width=0.4mm] (1.8,1.9)--(2.1,0);
\draw[dashed,line width=0.4mm] (0,1.9)--(1.8,1.9) node at (-0.6,1.9) {\small $1/5$};

\draw[-,line width=0.8mm](2.25+0.3,0.2)--(2.25+0.3,-0.2)node at (2.25+0.3,-0.75){\large $\frac{1}{4}$};
\draw[line width=0.4mm] (2.25+0.3,2.4)--(2.9,0);
\draw[dashed,line width=0.4mm] (0,2.45)--(2.25+0.3,2.45) node at (-0.6,2.45) {\small $1/4$};

\draw[-,line width=0.8mm](3+0.4,0.2)--(3+0.4,-0.2)node at (3+0.4,-0.75){\large $\frac{1}{3}$};
\draw[line width=0.4mm] (3+0.4,3.32)--(3.9,0);
\draw[dashed,line width=0.4mm] (0,3.3)--(3+0.4,3.3) node at (-0.6,3.3) {\small $1/3$};

\draw[-,line width=0.8mm](4.5+0.5,0.2)--(4.5+0.5,-0.2)node at (4.5+0.5,-0.75){ \large $\frac{1}{2}$};
\draw[line width=0.4mm] (4.5+0.5,4.9)--(5.7,0);
\draw[dashed,line width=0.4mm] (0,4.9)--(4.5+0.5,4.9) node at (-0.6,4.9) {\small $1/2$};

\draw[line width=0.4mm] (9,9)--(9,0);
\draw[dashed,line width=0.4mm] (0,9)--(9,9) node at (-0.6,9) {\small $1$};

`

\draw [->,black,line width=0.8mm] (0+\dl,0) -- (10+\dl,0);
\draw [->,black,line width=0.8mm] (0+\dl,0) -- (0+\dl,10);
\draw node at (10+\dl,-0.8) {\large $p$};
\draw [-,line width=0.8mm] (0+\dl,0.3) -- (0+\dl,-0.3) node at (0+\dl,-0.8) { $\unp$};
\draw [-,line width=0.8mm]  node at (-1.0+\dl,10.5) {\large $\sum_{k=1}^K \alpha_k \prof_{k}(p)$};
\draw [-,line width=0.8mm] (9+\dl,0.3) -- (9+\dl,-0.3) node at (9+\dl,-0.8) {\large $1$};
\draw[dashed,line width=0.4mm] (0+\dl,9)--(9+\dl,9) node at (-0.7+\dl,9) {\small $1/5$};

\draw[-,line width=0.8mm](1.8+\dl,0.2)--(1.8+\dl,-0.2)node at (1.8+\dl,-0.75){\small $\frac{1}{5}$};
\draw[-,line width=0.8mm](2.25+0.3+\dl,0.2)--(2.25+0.3+\dl,-0.2)node at (2.25+0.3+\dl,-0.75){\small $\frac{1}{4}$};
\draw[-,line width=0.8mm](3+0.4+\dl,0.2)--(3+0.4+\dl,-0.2)node at (3+0.4+\dl,-0.75){\small $\frac{1}{3}$};
\draw[-,line width=0.8mm](4.5+0.5+\dl,0.2)--(4.5+0.5+\dl,-0.2)node at (4.5+0.5+\dl,-0.75){\small $\frac{1}{2}$};

\draw[-,line width=0.4mm](0+\dl,0)--(1.8+\dl,9);
\draw[-,line width=0.4mm](1.8+\dl,9)--(1.9+\dl,7.5);
\draw[-,line width=0.4mm](1.9+\dl,7.5)--(2.25+0.3+\dl,9);
\draw[-,line width=0.4mm](2.25+0.3+\dl,9)--(2.8+\dl,6.5);
\draw[-,line width=0.4mm](2.8+\dl,6.5)--(3+0.4+\dl,9);
\draw[-,line width=0.4mm](3+0.4+\dl,9)--(3.8+\dl,5);
\draw[-,line width=0.4mm](3.8+\dl,5)--(4.5+0.5+\dl,9);
\draw[-,line width=0.4mm](4.5+0.5+\dl,9)--(5.8+\dl,4);
\draw[-,line width=0.4mm](5.8+\dl,4)--(9+\dl,9);

\node at (5,-2.2){\large \textbf{(a)}};
\node at (5+\dl,-2.2){\large \textbf{(b)}};
\end{tikzpicture}}
\caption{ Example for the construction of concave profit functions with
different supports, as in the proof of Proposition \protect\ref%
{prop:no-com-sup} for $K=5$. In \textbf{(a)} we illustrate the profit
functions for each segment where the value and optimal per-segment price
decay as $1/k$. In \textbf{(b)} we show $\sum_{k=1}^{K}\protect\alpha 
_{k}\prof_{k}(p)$ which is maximized at any of the per-segment optimal prices and
is bounded above by $1/K$.}
\label{fig:no-common-support}
\end{figure}
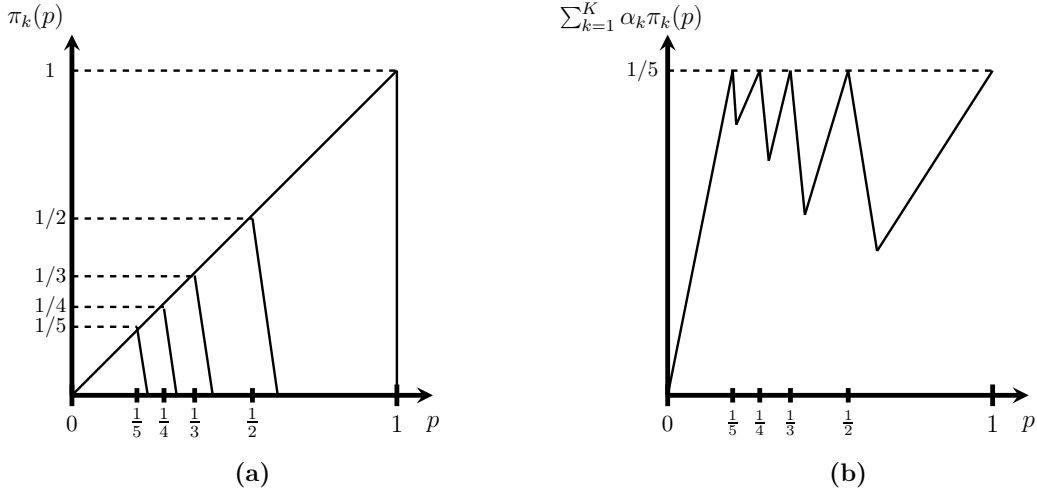
Proposition \ref{prop:no-com-sup} shows the importance of the common support
assumption. Above we construct concave profit functions that have
finite support, but the endpoints of the supports are increasing, see
Figure \ref{fig:no-common-support} \textbf{(a)}. In the construction, all
segments are given the same weight and the profit functions are triangle-shaped. All of them start at zero, go up along the 45 degree line, peak at $%
1/k$, and then go down sharply such that the upper end of the support of segment 
$k$ is strictly between $\frac{1}{k}$ and $\frac{1}{k-1}$ (solid lines in
Figure \ref{fig:no-common-support} \textbf{(a)}). In turn, $\Pi^\star$
(normalized by the per-segment proportions) grows logarithmically with $K$.
Since the upper bound of the supports are strictly increasing and non-overlapping,
the uniform price profit at the per-segment optimal prices, $\Pi ^{U}(1/k)$,
is constant and equal to 1 (normalized by the per-segment proportions), see
Figure \ref{fig:no-common-support} \textbf{(b)}. For example, consider $%
p_{3}^{\star }=1/3$. At this price, $\prof_{1}(1/3)=\prof_{2}(1/3)=0$ and $%
\prof_{3}(1/3)=\prof_{4}(1/3)=\prof_{5}(1/3)=1/3$. Hence $5\cdot \Pi ^{U}=0+0+3\cdot 
\frac{1}{3}=1$. As a result, the ratio in \eqref{eq:min-ratio} goes to zero as
the number of segments increases. Note that this only works because the
profit functions do not have common support. The non-common support allows
for the possibility of having profit functions such that at the per-segment
optimal prices some of them have zero profit.
We note that a similar version of this result was stated in %
\citet{malueg2006bounding}, Proposition 2. Their proof is inductive whereas our proof is
constructive and provides a transparent view regarding the significance of the common support.
Finally, we note that in the context of simple optimal auctions,  \citet{hartline2009simple} developed a related construction to argue the necessity of a single-item setting and regular distributions to obtain 
 revenue guarantees for anonymous reserve price 
versus the optimal auction that do not scale with the number of bidders. 

\subsection{Non-Concave Environments}\label{sec:non-con-env}

One of the most commonly analyzed families of distributions in the mechanism
design literature is regular distributions. These are distributions such
that the virtual value function is non-decreasing. Formally, we say that
distribution $F$ is regular if and only if $\phi (p)$ is non-decreasing: 
\begin{equation*}
\phi (p)\triangleq p-\frac{1-F(p)}{f(p)}.
\end{equation*}
As pointed out in Section \ref{lit}, several approximation guarantees have been obtained for these distributions in diverse settings. One of the
main insights used in the literature is that the profit function associated
with this family of distributions is concave in the quantile space. Indeed, let 
$\prof(p)=p\cdot (1-F(p))$ and consider the change of variables $q=1-F(p)$.
Define the profit function in the quantile space as $\widehat{\prof}(q)=q\cdot
F^{-1}(1-q)$. Then 
\begin{equation*}
\frac{d}{dq}\widehat{\prof}(q)=F^{-1}(1-q)-\frac{q}{f(F^{-1}(1-q))}=\phi
(F^{-1}(1-q)).
\end{equation*}%
and since $\phi (\cdot )$ is non-decreasing, we can conclude that $\widehat{\prof}(q)$ is
concave. The concavity of $\widehat{\prof}(q)$ allows arguments similar to the ones
employed in the triangular proof of Theorem \ref{thm1}. For example, %
\citet{dhangwatnotai2015revenue} use this property to show that with
one bidder, the expected profit from random pricing (uniformly selecting a
quantile) is half the profit of the optimal monopoly price. 
This suggests that a similar approach may work in our framework. In
particular, we ask whether it is possible with regular distributions
to exploit the concavity of the profit functions to obtain 
good approximation guarantees.

Recall that for each segment $k\geq 1$, the profit function comes from a cdf 
$F_{k}$, for which we assume its pdf $f_{k}$ is well-defined%
. To switch to the quantile space, for any $p\in \Theta $, we would
need to define 
\begin{equation*}
q_{k}=1-F_{k}(p)\quad \text{and}\quad \widehat{\prof}_{k}(q)=q\cdot F_{k}^{-1}(q).
\end{equation*}%
Let $\boldsymbol{q}=\{q_k\}_{k=1}^K$. The optimal uniform price profit, $\Pi ^{U}$, is given by 
\begin{align*}
\Pi ^{U}& =\max_{0\leq \boldsymbol{q}\leq 1}\sum_{k=1}^{K}\alpha _{k}\cdot 
\widehat{\prof}_{k}(q_{k}) \\
& \quad \text{s.t.}\quad F_{k}^{-1}(1-q_{k})=F_{j}^{-1}(1-q_{j})\quad
\forall k,j.
\end{align*}%
Note that in this formulation %
 the objective function is
the sum of concave functions. However, we have additional
constraints compared to the original formulation of $\Pi ^{U}$ in Section %
\ref{sec:mod}. These constraints stem from the fact that under uniform
pricing each segment receives the uniform price $p$, and since $%
q_{k}=1-F_{k}(p)$ we must have that $F^{-1}(1-q_{k})=F^{-1}(1-q_{j})$ for
all segments $k,j\geq 1$. At this point, the natural approach would be to
lower bound each $\widehat{\prof}_{k}$ by a triangle-shaped function---similar to
Figure \ref{fig:density1} \textbf{(a)}, but in the quantile space. We would then solve
the resulting optimization problem and, hopefully, obtain a good
approximation guarantee. 

Unfortunately, for regular distributions, in general,
the former approach fails.
Consider the case of triangular instances in
quantile space. These are instances for which the profit functions in the
quantile space are triangle-shaped---they have corresponding distributions that are regular. They are widely used in the literature
of approximate mechanism design as a bridge to provide good profit
guarantees. However, in our setting they can perform arbitrarily poorly. 

\begin{proposition}[\textbf{Triangular instances and failure of regular distributions}]\label{prop:triangle}{\ \\}
For triangular instances defined by 
\begin{equation*}
F_k(p) = \twopartdef{1,}{p\geq v_k;}{\frac{p\cdot(1-q_k)}{p\cdot(1-q_k)+v_k\cdot q_k},}{p<v_k,}
\end{equation*}
there exists a choice of $\{\alpha_k\}_{k=1}^K\in(0,1)^K$, $\{v_k\}_{k=1}^K\in\R^K_+$ and $\{q_k\}_{k=1}^K\in(0,1)^K$ that
delivers an
arbitrarily small \fcadd{profit ratio}  as the number of segments increases.
\end{proposition}%
\begin{proof}
Let us start by considering triangular instances. The profit functions are
\begin{equation*}
\prof_k(p)=p \cdot (1-F_k(p) )= \twopartdef{0,}{v\geq v_k;}{\frac{p\cdot v_k\cdot q_k}{p\cdot(1-q_k)+v_k\cdot q_k},}{v<v_k.}
\end{equation*}
Note that $\prof_k(p)$ is  increasing and concave up to $v_k$ and then is constant and equal to zero for $p\geq v_k$. For each curve the optimal price is $v_k$ (minus small $\eps>0$), and thus 
\begin{equation}
\Pi^\star= \sum_{k=1}^K\alpha_k \cdot v_k\cdot q_k.
\end{equation}

For the uniform price, the optimal price must be achieved at one of the $v_1,\dots,v_K$. Therefore, 
\begin{equation*}
\Pi^{U}= \max_{i\in\{1,\dots,K\}}\Big\{\sum_{k=i}^K \alpha_k\cdot \frac{v_i\cdot v_k\cdot q_k}{v_i \cdot(1-q_k)+v_k\cdot q_k}\Big\}.
\end{equation*}
Next we establish that $\Pi^{U}/\Pi^\star\rightarrow 0$ as $K\rightarrow \infty.$ Consider the instance 
\begin{equation*}
v_k= \frac{1}{(K-k+1)}, \quad q_k=0.5,\quad \text{and}\quad  \q_k=\frac{1}{K}, \quad \forall k\in\{1,\dots,K\}.
\end{equation*}
Hence, 
\begin{equation*}
\sum_{k=i}^K \alpha_k\cdot \frac{v_i\cdot v_k\cdot q_k}{v_i \cdot(1-q_k)+v_k\cdot q_k} =
\frac{1}{K}\sum_{k=i}^K  \frac{\frac{1}{(K-i+1)}\cdot \frac{1}{(K-k+1)}}{\frac{1}{(K-i+1)}+\frac{1}{(K-k+1)}} =
\frac{1}{K}\sum_{k=i}^K \frac{1}{2(K+1)-(k+i)}.
\end{equation*}
The last term above is decreasing in $i$, and therefore
\begin{equation*}
\Pi^{U}=\frac{1}{K}\sum_{k=1}^K \frac{1}{2K+1-k}=\frac{1}{K}\sum_{k=1}^K \frac{1}{K+k}\approx
\frac{1}{K} \int_{1}^K\frac{1}{K+x}dx=\frac{1}{K}\log\Big(\frac{2K}{K+1}\Big).
\end{equation*}
We also have that 
\begin{equation*}
\Pi^\star= \sum_{k=1}^K\alpha_k \cdot v_k\cdot q_k= \frac{1}{2K}\sum_{k=1}^K  \frac{1}{(K-k+1)}
= \frac{1}{2K}\sum_{k=1}^K  \frac{1}{k}\approx \frac{1}{2K}\log(K).
\end{equation*}
Thus, 
\begin{equation*}
\frac{\Pi^{U}}{\Pi^\star}\approx \frac{\frac{1}{K}\log\Big(\frac{2K}{K+1}\Big)}{\frac{1}{2K}\log(K)}
=2\frac{\log\Big(\frac{2K}{K+1}\Big)}{\log(K)}\rightarrow 0 \quad (\approx 2\cdot\log(2)/\log(K)).
\end{equation*}
\end{proof}

Given that for general regular distributions we may obtain arbitrarily poor guarantees, we next investigate if such results can be improved upon by considering a commonly used sub-family of distributions with more structure.  We consider distributions with monotone (non-increasing) inverse hazard rate $(1-F_k(p))/f_k(p)$. Note that triangular distributions do not belong to this family as they an have increasing hazard rate. 
Interestingly, 
even  distributions with monotone hazard rate 
can deliver only an arbitrarily small profit guarantee.

\begin{proposition}[\textbf{Failure of monotone hazard rate distributions}]\label{prop:fail-regu}{\ \\} There exist distributions with monotone inverse hazard rate and common and bounded support for which the optimal uniform price delivers an arbitrarily small \fcadd{profit ratio} as the number of segments increases.
\end{proposition}

\begin{proof} We construct regular distributions $\{F_k\}_{k=1}^K$ such that $\Pi^{U}/\Pi^\star\rightarrow 0$ as $K\uparrow \infty$. For $L>0$ large, define
\begin{equation*}
F_k(p)=\frac{1-e^{-(K-k+1)p}}{1-e^{-(K-k+1)L}} \quad \forall p\geq0,\quad\text{and}\quad \alpha_k=1/K\quad \forall k\in\{1,\dots,K\}.
\end{equation*}
Thus, we consider truncated exponential distributions with support in $[0,L]$ 
. Note that these distributions have monotone inverse hazard rate because 
\begin{equation*}
\frac{1-F_k(p)}{f_k(p)}=
\frac{1}{K-k+1} - \frac{1}{(K-k+1)}\cdot\frac{e^{-(K-k+1)L}}{e^{-(K-k+1)p}}
\end{equation*}
is non-increasing.
 For all $k\ge 1$, the profit functions are 
\begin{equation*}
\prof_k(p)=p\cdot (1-F_k(p))= p\cdot \frac{e^{-(K-k+1)p}-e^{-(K-k+1)L}}{1-e^{-(K-k+1)L}},
\end{equation*}
whereas the per-segment optimal prices satisfy 
\begin{equation*}
p_k^\star = \frac{1-e^{-(K-k+1)(L-p_k^\star)}}{K-k+1},\quad \text{and}\quad \prof_k(p_k^\star) =  \frac{1-e^{-(K-k+1)(L-p_k^\star)}}{K-k+1}\cdot \frac{e^{-(K-k+1)p_k^\star}-e^{-(K-k+1)L}}{1-e^{-(K-k+1)L}}.
\end{equation*}
First, notice that $p_k^\star \leq 1/(K-k+1)$. Then the third-degree price discrimination profit can be bounded as follows:
\begin{flalign*}
\Pi^\star &=\sum_{k=1}^K\alpha_k\prof_k(p_k^\star)\\
&=\sum_{k=1}^K\frac{1}{K}\frac{e^{(K-k+1)p_k^\star}\left(e^{-(K-k+1)p_k^\star}-e^{-(K-k+1)L}\right)^2}{(K-k+1)(1-e^{-(K-k+1)L})}\\
&\stackrel{(a)}{\geq}\sum_{k=1}^K\frac{1}{K}\frac{e^{1}\left(e^{-1}-e^{-(K-k+1)L}\right)^2}{K-k+1}\\
&\stackrel{(b)}{\geq}
\frac{e\left(e^{-1}-e^{-L}\right)^2}{K}\sum_{k=1}^K\frac{1}{K-k+1}\\
&\approx \frac{e\left(e^{-1}-e^{-L}\right)^2}{K}\cdot \log(K),
\end{flalign*}
where in $(a)$ we used the fact that the function $e^{\lambda x}\left(e^{-\lambda x}-e^{\lambda L}\right)^2$ is decreasing for $x\in[0,L]$ and that $p_k^\star<L$. In $(b)$ we used $L$ large enough such that $L>1$.
The uniform price profit for some price $p$ is 
\begin{flalign*}
\sum_{k=1}^K\alpha_k\prof_k(p)&=\frac{1}{K}\sum_{k=1}^K p\cdot \frac{e^{-(K-k+1)p}-e^{-(K-k+1)L}}{1-e^{-(K-k+1)L}}\\
&\leq\frac{p}{K}\sum_{k=1}^K  e^{-(K-k+1)p}\\
&=\frac{p}{K}\sum_{k=1}^K  e^{-kp}\\
&=\frac{p}{K}\cdot\frac{1-e^{-Kp}}{e^p-1}\\
&\leq\frac{1-e^{-Kp}}{K}\\
&\leq \frac{1}{K},
\end{flalign*}
where the second to last inequality holds because we always have that $p+1\leq e^p$.
With this, we can conclude that 
\begin{equation*}
\frac{\Pi^{U}}{\Pi^\star} = \frac{\max_{p\geq 0}\Big\{\sum_{k=1}^K\alpha_k\prof_k(p)\Big\}}{\Pi^\star}
\leq \frac{\frac{1}{K}}{\frac{e\left(e^{-1}-e^{-L}\right)^2}{K}\cdot \log(K)}= \frac{1}{e\left(e^{-1}-e^{-L}\right)^2\cdot \log(K)}\rightarrow 0,\quad K\uparrow\infty.
\end{equation*}
\end{proof}Proposition \ref{prop:fail-regu} establishes that for some
monotone hazard rate distributions, uniform pricing can perform arbitrarily poorly
compared to optimal third-degree price discrimination.
The intuition behind this
result is similar to that of Proposition \ref{prop:no-com-sup}. We consider
exponential distributions such that at the optimal uniform price, most of
the per-segment profits will be low, and therefore they will not contribute
much to $\Pi ^{U}$, see Figure \ref{fig:no-common-support-2}. Since for
exponentials, the associated profit functions decay quickly after they peak,
they behave in a similar manner as the case of non-common support
distributions where the upper end of the support is increasing. Indeed, in
the proof of Proposition \ref{prop:fail-regu} we obtain a profit
guarantee similar to that found in the proof of Proposition \ref{prop:no-com-sup}, namely, $%
O(1/\log (K))$. Finally, we note that to prove the proposition,
we use truncated exponential distributions. In turn, we only relax the concavity 
of the profit functions, but we keep the common and finite support assumptions intact.

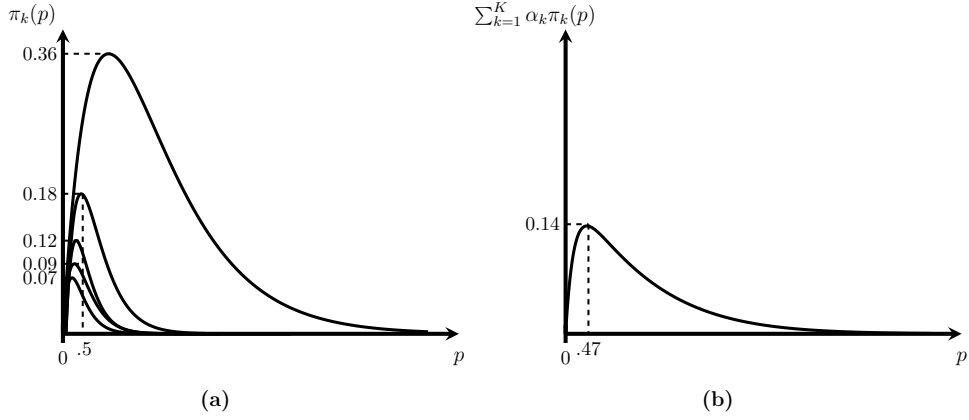
\begin{figure}[h]
\centering
\scalebox{0.675}{\begin{tikzpicture}[baseline=0pt,scale=0.6]
\def\n{9}
\def\dl{16.5}

\draw [->,black,line width=0.8mm] (0,0) -- (13,0);
\draw [->,black,line width=0.8mm] (0,0) -- (0,10);
\draw node at (13,-0.8) {\large $p$};
\draw [-,line width=0.8mm] (0,0.3) -- (0,-0.3) node at (0,-0.8) { $\unp$};
\draw [-,line width=0.8mm]  node at (-1.0,10.5) {\large $\prof_{k}(p)$};

\draw[dashed,line width=0.4mm](0,9.2)--(1.75,9.2) node at (-0.75,9.2) {0.36};
\draw[dashed,line width=0.4mm](0,4.6)--(0.7,4.6) node at (-0.75,4.6) {0.18};
\draw[dashed,line width=0.4mm](0,3.06)--(0.7,3.06) node at (-0.75,3.06) {0.12};
\draw[dashed,line width=0.4mm](0,2.3)--(0.6,2.3) node at (-0.75,2.3) {0.09};
\draw[dashed,line width=0.4mm](0,1.83)--(0.2,1.83) node at (-0.75,1.83) {0.07};

\draw[dashed,line width=0.4mm](0.65,4.6)--(0.65,0) node at (0.7,-0.5) {.5};

\draw[black,line width=0.6mm,shift={(0,0)}] plot file {regular-1.data};
\draw[black,line width=0.6mm,shift={(0,0)}] plot file {regular-2.data};
\draw[black,line width=0.6mm,shift={(0.1,0)}] plot file {regular-3.data};
\draw[black,line width=0.6mm,shift={(0.1,0)}] plot file {regular-4.data};
\draw[black,line width=0.6mm,shift={(0,0)}] plot file {regular-5.data};

\draw [->,black,line width=0.8mm] (0+\dl,0) -- (13+\dl,0);
\draw [->,black,line width=0.8mm] (0+\dl,0) -- (0+\dl,10);
\draw node at (13+\dl,-0.8) {\large $p$};
\draw [-,line width=0.8mm] (0+\dl,0.3) -- (0+\dl,-0.3) node at (0+\dl,-0.8) { $\unp$};
\draw [-,line width=0.8mm]  node at (-1.0+\dl,10.5) {\large $\sum_{k=1}^K \alpha_k \prof_k(p)$};

\draw[black,line width=0.6mm,shift={(0+\dl,0)}] plot file {regular-T.data};
\draw[dashed,line width=0.4mm](0+\dl,3.6)--(0.75+\dl,3.6) node at (-0.75+\dl,3.6) {0.14};
\draw[dashed,line width=0.4mm](0.75+\dl,3.6)--(0.75+\dl,0) node at (0.75+\dl,-0.5) {.47};

\node at (5,-2.2){\large \textbf{(a)}};
\node at (5+\dl,-2.2){\large \textbf{(b)}};
\end{tikzpicture}}
\caption{ Example for the construction of profit functions from regular
distributions in Proposition \protect\ref{prop:fail-regu} with $K=5$. In 
\textbf{(a)} we illustrate the profit functions for each segment where the
per-segment optimal profit is $e^{-1}/k$ and the per-segment optimal price $%
1/k$. In \textbf{(b)} we show $\sum_{k=1}^{K}\protect\alpha _{k}\prof_{k}(p)$
with maximum value $0.14$ (bounded above by $1/K$).}
\label{fig:no-common-support-2}
\end{figure}

\section{Worst-Case Performance}\label{sec:wcp}

In this brief section, our objective is to investigate \fcadd{how the worst case
performance of uniform pricing depends on the number of segments}. For this purpose, suppose there are $K$
segments and without loss of generality let us assume that 
\begin{equation}
0\leq \alpha _{1}\prof_{1}(p_{1}^{\star })\leq \alpha _{2}\prof_{2}(p_{2}^{\star
})\leq \cdots \leq \alpha _{K}\prof_{K}(p_{K}^{\star }).
\label{eq:order-reven-w-c}
\end{equation}%
Using this condition we can verify that the ratio $\Pi ^{U}/\Pi^\star$ is
always bounded below by $1/K$. Indeed, note that $\Pi ^{U}\geq \Pi
^{U}(p_{K}^{\star })$ and 
\begin{equation*}
\Pi ^{U}(p_{K}^{\star })=\sum_{k=1}^{K}\q_{k}\prof_{k}(p_{K}^{\star })\geq \q%
_{K}\prof_{K}(p_{K}^{\star })=\frac{\overbrace{\q_{K}\prof_{K}(p_{K}^{\star
})+\cdots +\q_{K}\prof_{K}(p_{K}^{\star })}^{K\text{ times}}}{K}\geq \frac{1}{K}%
\sum_{k=1}^{K}\alpha _{k}\prof_{k}(p_{k}^{\star }),
\end{equation*}%
where in the last inequality  we use \eqref{eq:order-reven-w-c}. 
This proves that the worst-case performance of uniform pricing with respect
to third-degree price discrimination is  $1/K$. In what follows,
we argue that this lower bound performance can indeed be achieved.\footnote{%
Proposition 2 in \citet{malueg2006bounding} provides an inductive argument for this result
with linear demand functions. Here we provide a constructive argument with
atomic distributions.}

\fcadd{We now suppose that the demand in every segment $k$ is described by a Dirac distribution with an atom at value $v_k$. 
We denote by $\mathcal{D}$ the set of the segmentations and per segment demand functions that are generated by Dirac distributions, thus $\{\alpha_k,v_k\}_{k=1}^K\in\mathcal{D}$.
}

 Under perfect price discrimination, the
monopolist can charge the price $p_{k} = v_k$ to the buyer and extract full
surplus: $\sum_{k=1}^{K}\q_{k}v_{k}$. Under uniform pricing, the monopolist
charges a fixed price $p$ across all segments and collects profit only from
those segments whose value, $v_{k}$, is larger than $p$: $\sum_{k=1}^{K}\q%
_{k}p\1_{v_{k}\geq p}$. In turn, the optimization problem we would like to
solve to assess the performance of uniform pricing becomes 
\begin{equation*}
\inf_{\alpha _{k},v_{k}}\left\{ \frac{\max_{p\geq 0}\sum_{k=1}^{K}\q_{k}p\1%
_{v_{k}\geq p}}{\sum_{k=1}^{K}\q_{k}v_{k}},\quad \text{s.t.}\quad
\sum_{k=1}^{K}\alpha _{k}=1,\:\:\alpha _{k}\geq 0\:\:\forall k\right\} .
\end{equation*}%
Observe that in the above problem, without loss of generality, we can assume
that the values $v_{k}$ are ordered. This allows us to simplify the
numerator in the objective above. Note that the maximum in the numerator
must be achieved at some price $p_{j}=v_j$ for $j\in \{1,\dots ,K\}$ and for any 
$p_{j}$ we have 
\begin{equation*}
\sum_{k=1}^{K}\q_{k}p_{j}\1_{p_{k}\geq p_{j}}=p_{j}\sum_{k=j}^{K}\alpha _{k}.
\end{equation*}%
In turn, this enables us to reformulate the problem as 
\begin{equation}
\inf_{\alpha _{k},p_k}\left\{ \frac{\max_{j\in \{1,\dots
,K\}}p_{j}\sum_{k=j}^{K}\alpha _{k}}{\sum_{k=1}^{K}\q_{k}p_{k}},\quad \text{%
s.t.}\quad p_{1}\leq \cdots \leq p_{K},\quad \sum_{k=1}^{K}\alpha
_{k}=1,\:\:\alpha _{k}\geq 0\:\:\forall k\right\} .  \label{eq:reform-w-c}
\end{equation}%
Next, we exhibit values of $p_{k}$ and $\alpha _{k}$ such that the ratio in
the problem above is arbitrarily close to the lower bound $1/K$. Let $\eps>0$
be small, and define the prices 
\begin{equation*}
p_{k}=\frac{\eps}{K}\left( \frac{1+\eps}{\eps}\right) ^{k},\quad k\in
\{1,\dots ,K-1\},\quad \text{and }\quad p_{K}=\frac{1}{K}\left( \frac{1+\eps%
}{\eps}\right) ^{K-1},
\end{equation*}%
and let the per-segment proportions be 
\begin{equation*}
\alpha _{k}=\frac{1}{K}\frac{1}{p_{k}},\quad k\in \{1,\dots ,K\}.
\end{equation*}%
Next, we verify that the above prices and proportions are feasible. Indeed,
it is easy to verify that $p_{k}$ is increasing in $k$, while for the
per-segment proportions we have that $\alpha _{k}>0$ and 
\begin{flalign*}
\sum_{k=1}^K\alpha_k&= \sum_{k=1}^{K-1}\frac{1}{\eps}\left(\frac{\eps}{1+\eps}\right)^k + \left(\frac{\eps}{1+\eps}\right)^{K-1}=1.
\end{flalign*}Now let us look at the objective in problem %
\eqref{eq:reform-w-c}. Given our choice of prices and proportions, we have
that $\alpha _{k}p_{k}=1/K$, and therefore the denominator in %
\eqref{eq:reform-w-c} equals 1. For the numerator consider $j\in \{1,\dots
,K-1\}$. Then 
\begin{flalign*}
p_j\sum_{k=j}^K\alpha_k &= \frac{\eps}{K}\left(\frac{1+\eps}{\eps}\right)^j
\left(
 \sum_{k=j}^{K-1}\frac{1}{\eps}\left(\frac{\eps}{1+\eps}\right)^k + \left(\frac{\eps}{1+\eps}\right)^{K-1}
\right)\\
&= \frac{\eps}{K}\left(\frac{1+\eps}{\eps}\right)^j
\left(
 \frac{\eps}{1+\eps}
\right)^{j-1}\\
&=\frac{1+\eps}{K},
\end{flalign*}and $p_{K}\alpha _{K}=1/K$. Therefore the objective in %
\eqref{eq:reform-w-c} evaluated at our current choice of prices and
proportions equals 
\begin{equation*}
\max_{j\in \{1,\dots ,K\}}\left\{ p_{j}\sum_{k=j}^{K}\alpha _{k}\right\}
=\max \left\{ \frac{1}{K},\frac{1+\eps}{K}\right\} =\frac{1+\eps}{K}.
\end{equation*}%
In conclusion, we have exhibited an instance for which the performance of
the optimal uniform price is arbitrarily close  to the worst
performance guarantee $1/K$. We summarize this result in the following
proposition. 
\begin{proposition}[\textbf{Worst performance achieved}]{\ \\}
\fcadd{The worst case profit ratio in the class of Dirac demand functions is given by}
\begin{equation*}
\inf_{\{\alpha_k,v_k\}_{k=1}^K\in \mathcal{D}}\frac{\Pi^{U}}{\Pi^\star}
= \frac{1}{K}.
\end{equation*}
\end{proposition}

\section{Connection to Sequential Screening}\label{revision-seq-screening}

The result in Theorem \ref{thm1} is intimately related to the
problem of ex-post individually rational sequential screening, in which the seller must
optimally design a menu of contracts that incentivize buyers of different
types to self-select, see e.g., \citet{krst15} and \citet{bergemann2017scope}.
The connection between the two settings comes from considering the types in sequential screening as the segments of our paper and observing that optimal static pricing in the screening setting is the same as our
optimal uniform pricing. Additionally,  the optimal screening profit is bounded above by the optimal perfect price discrimination profit. As we show next, this allows us to prove a half approximation result for the sequential screening setting.

In what follows we formally introduce the sequential screening setting and then establish the connection with our third-degree price discrimination setting. There is a seller selling one unit of an object at zero cost to a buyer with an outside option of zero. The buyer is of type 
$k\in \{1,\dots ,K\}$, with probability $\alpha _{k}$, $\alpha _{k}>0$ and $
\sum_{k=1}^{K}\alpha _{k}=1$.
Both parties are
risk-neutral and have quasilinear utility functions. There are two periods. In the first period, the buyer privately learns her type $k$---which  provides information about her true willingness-to-pay for the object---and then the parties contract. The contract specifies allocation and payment as a function of reported
interim type and ex-post value. In the second period, the buyer privately
learns her value $\theta$---drawn from a distribution function $F_{k}(\theta)$ with density function $f_{k}(\theta)$---and allocations and transfers are realized. 

We consider direct revelation mechanisms, with
allocations $x_{k}:\Theta \rightarrow \lbrack 0,1]$ and transfers $%
t_{k}:\Theta \rightarrow \mathbb{R}$, that depend on reported interim type $%
k^{\prime }$ and ex-post value $\theta^{\prime }$. Then  the ex-post utility of
 a buyer who reported $k$ in the
first period and $v ^{\prime }$ in the second period while her true
value is $v $ is given by:
\[
u_{k}(\theta ;\theta^{\prime })\triangleq v \cdot x_{k}(\theta
^{\prime })-t_{k}(\theta^{\prime }).
\]%
Similarly, the interim expected utility of a buyer whose true
interim type is $k$, but reported to the mechanism $k^{\prime}$\label{pg:dou-dev} and is truthful in the second period, is given by: 
\[
U_{kk^{\prime }}\triangleq \int_{\Theta }u_{k^{\prime }}(z;z)\cdot f_{k}(z)dz.
\]

There are two kinds of incentive compatibility constraints that must be
satisfied. The first is the ex-post incentive compatibility
constraint which requires that for any report in the first
period, truth-telling is optimal in the second period: 
\begin{equation*}
u_{k}(\theta,\theta )\geq u_{k}(\theta ;\theta ^{\prime })\quad \forall k\in
\{1,\dots ,K\},\forall \theta \in \Theta .  %
\end{equation*}%
The second is the interim incentive compatibility constraint 
which requires that truth-telling is optimal in the first period: 
\begin{equation*}
U_{kk}\geq U_{kk^{\prime }}\quad \forall k,k^{\prime }\in \{1,\dots ,K\}. 
\end{equation*}%
Finally, the key constraint in the sequential screening problem with ex-post participation constraints is that the buyer must be willing to participate after having learned her type and value:
\begin{equation*}
u_{k}(\theta )\geq 0,\quad \forall k\in \{1,\dots ,K\},\quad \forall \theta
\in \Theta.%
\end{equation*}%
The seller aims to maximize the expected transfers from a mechanism that satisfies both incentive compatibility constraints and the ex-post participation constraint. Lemma 1 in \citet{bergemann2017scope} implies the following reformulation of the seller's problem in which we only need to solve for the allocations and the utility of the lowest ex-post type:
\begin{flalign*}%
&\quad\Pi^{seq}\triangleq \max_{0\leq x_k\leq 1,u_k} \quad 
-\sum_{k=1}^{K}\q_ku_{k}+\sum_{k=1}^{K}\q_k\int_\Theta x_{k}(z)\left(z-\frac{1-F_k(z)}{f_k(z)}\right)f_{k}(z )dz\nonumber\\
&\qquad  \text{s.t}\quad x_{k}(\cdot) \quad \text{is non-decreasing},\quad \forall k\in\{1,\dots,K\}\nonumber\\
&\qquad \qquad u_{k}\geq 0,\quad \forall k\in\{1,\dots,K\}\nonumber\\
&\qquad \qquad u_{k}+ 
\int_\Theta x_{k}(z)(1-F_k(z))dz\geq u_{k'}+\int_\Theta x_{k'}(z)(1-F_k(z))dz,\quad
\forall k,k'\in\{1,\dots,K\}\nonumber.
\end{flalign*}
The first set of constraints comes from the ex-post incentive compatibility constraints in the original formulation. The second set of constraints ensures ex-post participation. The final constraints come from the interim incentive compatibility constraints. 

To see the connection between the setting above and our third-degree price discrimination setting, first consider $\Pi^{seq}$ without the interim incentive compatibility constraints. In this case the problem decouples across types, and for each type it reduces to:
\begin{flalign*}
&\quad \max_{0\leq x\leq 1} \quad 
\int_\Theta x(z)\cdot \q_k\left(z-\frac{1-F_k(z)}{f_k(z)}\right)f_{k}(z)dz\nonumber\\
&\qquad  \text{s.t}\quad x(\cdot) \quad \text{is non-decreasing}. \nonumber
\end{flalign*}
Note that $u_k=0$ for all $k$ is optimal in all decoupled problems. 
It is well known (see e.g., \cite{rize83}) that the problem above has a bang-bang solution, say $p_k$. Therefore, 
\begin{equation}\label{eq:rev-seq-upp}
\Pi^{seq}\leq \sum_{k=1}^K \int_{p_k}^{\ovp} \q_k\left(z-\frac{1-F_k(z)}{f_k(z)}\right)f_{k}(z)dz = 
\sum_{k=1}^K \q_k \prof_k(p_k)\leq \sum_{k=1}^K \q_k \prof_k(p_k^\star) =\Pi^\star.
\end{equation}
That is, the optimal solution in the sequential screening problem is bounded above by the optimal third-degree price discrimination solution.

Additionally, we can consider $\Pi^{seq}$ in the case where the seller uses a static price---a price that is the same regardless of the buyer's interim type. In other words, we set $x_k(\cdot)=x(\cdot)$ for all $k$. In this case,  after setting $u_k=0$ for all $k$, the interim incentive compatibility constraints are directly satisfied, and the problem becomes:
\begin{flalign}
&\quad \Pi^{static}\triangleq \max_{0\leq x\leq 1} \quad 
\int_\Theta x(z)\cdot\Big(\sum_{k=1}^{K}\q_k\left(z-\frac{1-F_k(z)}{f_k(z)}\right)f_{k}(z )\Big)dz\nonumber\\
&\qquad  \text{s.t}\quad x(\cdot) \quad \text{non-decreasing}. \nonumber
\end{flalign}
As before, the optimal solution is bang-bang:
\begin{equation}\label{eq:rev-seq-low}
\Pi^{static} = \max_{p\in \Theta} \int_{p}^{\ovp} \sum_{k=1}^{K}\q_k\left(z-\frac{1-F_k(z)}{f_k(z)}\right)f_{k}(z )dz\ =\max_{p\in \Theta}\sum_{k=1}^{K}\q_k\prof_k(p) =\Pi^U.
\end{equation}
That is, the static optimal solution in the sequential screening problem coincides with the optimal uniform price solution. 

Under the conditions of Theorem \ref{thm1}, conditions \eqref{eq:rev-seq-upp} and \eqref{eq:rev-seq-low} imply that:
\begin{equation*}
\frac{\Pi^{static}}{\Pi^{seq}} =\frac{\Pi^{U}}{\Pi^{seq}} \geq \frac{\Pi^{U}}{\Pi^\star} \geq 1/2.
\end{equation*}
We have thus established the following corollary of Theorem \ref{thm1}. 
\begin{corollary}[\textbf{Half approximation in sequential screening}]{\ \\} Suppose that the assumptions of Theorem \ref{thm1} hold. Then in the ex-post individually rational screening setting of \citet{krst15} and \citet{bergemann2017scope}, the optimal static contract delivers a 1/2-approximation for the seller's profits. 
\end{corollary}

\section{Conclusion}

We consider the profit performance of uniform pricing in settings where a
monopolist may engage in third-degree price discrimination. We establish
that, for concave profit functions with common support, using a single
price can achieve half of the optimal profit the monopolist could
potentially garner by engaging in third-degree price discrimination. Our
profit guarantee does not depend on the number of market
segments or prices the seller might use.  
The different arguments we provide for \cref{thm1} highlight that the monopolist  
can achieve the performance guarantee of 1/2 in different informational settings by using a simple uniform price.

We then investigate the scope of
our profit guarantee. We establish that by relaxing either the concavity or
the common support assumption, uniform pricing can lead to arbitrarily poor
profit guarantees as the number of market segments increases.
Interestingly, for regular distributions and triangular instances---leading
cases in the literature of approximate mechanism design---we show that
uniform pricing can again secure only very poor profit guarantees that scale with
the number of market segments.

Depending on the nature of a market, different types of price discrimination
are possible. A plausible direction for future work is to consider an
environment in which the seller can exercise second-degree price
discrimination by creating a menu of prices and quantities. For example,
consider a setting with $K$ different markets, each of which is characterized
by a different distribution of valuations. Within each market the seller may
offer an optimal menu of prices and quantities. However, due to legal or
business constraints
 such powerful
price discrimination may not be implementable. Instead, the seller might only be
able to offer the same menu across all markets. In turn, it becomes a
natural question to explore the performance of this limited menu versus the
full discriminating one. %

\bibliographystyle{econometrica}
\bibliography{general}

\newpage

\appendix

\end{document}